\documentclass[conference]{IEEEtran}
\usepackage{epsfig,epstopdf} 
\usepackage{amsmath,amssymb,amsthm,cite,graphicx,array}
\usepackage{algorithm}
\usepackage[noend]{algpseudocode}
\theoremstyle{definition}
\newtheorem{theorem}{Theorem}
\newtheorem{lemma}{Lemma}

\newtheorem{definition}{Definition}
\newtheorem{remark}{Remark}
\newtheorem{example}{Example}
\usepackage{color}

\title{Capacity of Some Index Coding Problems with Symmetric Neighboring Interference}

\begin{document}

\author{Mahesh~Babu~Vaddi~and~B.~Sundar~Rajan\\ 
 Department of Electrical Communication Engineering, Indian Institute of Science, Bengaluru 560012, KA, India \\ E-mail:~\{mahesh,~bsrajan\}@ece.iisc.ernet.in }
 
\maketitle
\begin{abstract}
A single unicast index coding problem (SUICP) with symmetric neighboring interference (SNI) has equal number of $K$ messages and $K$ receivers, the $k$th receiver $R_{k}$ wanting the $k$th message $x_{k}$ and having the side-information $\mathcal{K}_{k}=(\mathcal{I}_{k} \cup x_{k})^c,$ where ${I}_k= \{x_{k-U},\dots,x_{k-2},x_{k-1}\}\cup\{x_{k+1}, x_{k+2},\dots,x_{k+D}\}$ is the interference with $D$ messages after and $U$ messages before its desired message. Maleki, Cadambe and Jafar obtained the capacity of this symmetric neighboring interference single unicast index coding problem (SNI-SUICP) with $(K)$ tending to infinity and Blasiak, Kleinberg and Lubetzky for the special case of $(D=U=1)$ with $K$ being finite. In this work, for any finite $K$ and arbitrary $D$ we obtain the capacity for the case $U=\text{gcd}(K,D+1)-1.$ Our proof is constructive, i.e., we give an explicit construction of a linear index code achieving the capacity. 
\end{abstract}
\section{Introduction and Background}
\label{sec1}

\IEEEPARstart {A}{n} index coding problem, comprises a transmitter that has a set of $K$ independent messages, $X=\{ x_0,x_1,\ldots,x_{K-1}\}$, and a set of $M$ receivers, $R=\{ R_0,R_1,\ldots,R_{M-1}\}$. Each receiver, $R_k=(\mathcal{K}_k,\mathcal{W}_k)$, knows a subset of messages, $\mathcal{K}_k \subset X$, called its \textit{Known-set} or the \textit{side-information}, and demands to know another subset of messages, $\mathcal{W}_k \subseteq \mathcal{K}_k^\mathsf{c}$, called its \textit{Want-set} or \textit{Demand-set}. A naive technique would be to broadcast all the messages in $K$ time slots. Instead, the transmitter can take cognizance of the side-information of the receivers and broadcast coded messages, called the index code, over a noiseless channel. The objective is to minimize the number of coded transmissions, called the length of the index code, such that each receiver can decode its demanded message using its side-information and the coded messages.

The problem of index coding with side-information was introduced by Birk and Kol \cite{ISCO}. Ong and Ho \cite{OnH} classified the binary index coding problem depending on the demands and the side-information possessed by the receivers. An index coding problem is unicast if the demand sets of the receivers are disjoint. An index coding problem is single unicast if the demand sets of the receivers are disjoint and the cardinality of demand set of every receiver is one. Any unicast index problem can be converted into a single unicast index coding problem. A single unicast index coding problem (SUICP) can be described as follows: Let $\{x_{0}$,$x_{1}$,\ldots,$x_{K-1}\}$ be the $K$ messages, $\{R_{0}$,$R_{1},\ldots,R_{K-1}\}$ are $K$ receivers and $x_k \in \mathcal{A}$ for some alphabet $\mathcal{A}$ and $k=0,1,\ldots,K-1$. Receiver $R_{k}$ is interested in the message $x_{k}$ and knows a subset of messages in $\{x_{0}$,$x_{1}$,\ldots,$x_{K-1}\}$ as side-information. 


A solution (includes both linear and nonlinear) of the index coding problem must specify a finite alphabet $\mathcal{A}_P$ to be used by the transmitter, and an encoding scheme $\varepsilon:\mathcal{A}^{t} \rightarrow \mathcal{A}_{P}$ such that every receiver is able to decode the wanted message from the $\varepsilon(x_0,x_1,\ldots,x_{K-1})$ and the known information. The minimum encoding length $l=\lceil log_{2}|\mathcal{A}_{P}|\rceil$ for messages that are $t$ bit long ($\vert\mathcal{A}\vert=2^t$) is denoted by $\beta_{t}(G)$. The broadcast rate of the index coding problem with side-information graph $G$ is defined \cite{ICVLP} as,
\begin{align*}
\beta(G) \triangleq   \inf_{t} \frac{\beta_{t}(G)}{t}.
\end{align*}

If $t = 1$, it is called scalar broadcast rate. For a given index coding problem, the broadcast rate $\beta(G)$ is the minimum number of index code symbols required to transmit to satisfy the demands of all the receivers. The capacity $C(G)$ for the index coding problem is defined as the maximum number of message symbols transmitted per index code symbol such that every receiver gets its wanted message symbols and all the receivers get equal number of wanted message symbols. The broadcast rate and capacity are related as 
\begin{center}	
$C(G)=\dfrac{1}{\beta(G)}$.
\end{center}  

Instead of one transmitter and $K$ receivers, the SUICP can also be viewed as $K$ source-receiver pairs with all $K$ sources connected with all $K$ receivers through a common finite capacity channel and all source-receiver pairs connected with either zero of infinite capacity channels. This problem is called multiple unicast index coding problem in \cite{MCJ}.

In a symmetric neighboring interference single unicast index coding problem (SNI-SUICP) with equal number of $K$ messages and receivers, each receiver has interfering messages, corresponding to the $D$ messages after and $U$ messages before its desired message. In this setting, the $k$th receiver $R_{k}$ demands the message $x_{k}$ having the interference
\begin{equation}
\label{antidote}
{I}_k= \{x_{k-U},\dots,x_{k-2},x_{k-1}\}\cup\{x_{k+1}, x_{k+2},\dots,x_{k+D}\}.
\end{equation}

The side-information of this setting is given by
\begin{align}
\label{sideinformation}
\mathcal{K}_{k}=(\mathcal{I}_{k} \cup x_{k})^c.
\end{align}

Maleki \textit{et al.} \cite{MCJ} found the capacity of SNI-SUICP with $K\rightarrow \infty$ to be
\begin{align}
\label{capacityint}
C=\frac{1}{D+1}~\text{per~message}.
\end{align}

Also, it was shown in  \cite{MCJ} that  the outer bound for the capacity of SNI-SUICP for finite $K$ and is given by
\begin{align}
\label{outerbound}
C \leq \frac{1}{D+1}.
\end{align}
Blasiak \textit{et al.} \cite{ICVLP} found the capacity of SNI-SUICP with $U=D=1$ by using linear programming bounds to be $\frac{\left\lfloor \frac{K}{2}\right\rfloor}{K}$. 

Jafar \cite{TIM} established the relation between index coding problem and topological interference management problem. The SNI-SUICP is motivated by topological interference management problems. The capacity and optimal coding results in index coding can be used in corresponding topological interference management problems.

\subsection{Contributions}
The contributions of this paper are summarized below:
\begin{itemize}
\item We derive the capacity of SNI-SUICP with $D$ interfering messages after and $U=\text{gcd}(K,D+1)-1$ interfering messages before the desired message.
\item We show that AIR matrices of size $K \times (D+1)$ can be used as an encoding matrix to generate optimal index code over every field.
\end{itemize}

All the subscripts in this paper are to be considered $~\text{\textit{modulo}}~ K$. In the remaining paper, we refer SNI-SUICP with $D$ interfering messages after and $U=\text{gcd}(K,D+1)-1$ interfering messages before the desired message as SNI-SUICP.

The remaining part of this paper is organized as follows. In Section \ref{sec2} we define and review the properties of Adjacent Row Independent (AIR) matrices which is already discussed in detail in \cite{VaR2} in the context of optimal index codes with symmetric, neighboring consecutive side-information. Except for the proof of Lemma \ref{lemma1} this section is a slightly modified version available in \cite{VaR3} and is present only for the sake of being self-contained. In \ref{sec3}, we show that AIR matrix can be used as an encoding matrix to generate optimal index code for SNI-SUICP. We conclude the paper in Section \ref{sec4}.

%
\section{Review of AIR matrices}
\label{sec2}

In \cite{VaR2}, we gave the construction of AIR matrix and we used AIR matrices to give optimal length index codes for one-sided symmetric neighboring and consecutive side-information index coding problems (SNC-SUICP). In \cite{VaR1}, we constructed optimal vector linear index codes for two-sided SNC-SUICP. In \cite{VaR3}, we gave a low-complexity decoding for SNC-SUICP with AIR matrix as encoding matrix. The low complexity decoding method helps to identify a reduced set of side-information for each users with which the decoding can be carried out. By this method every receiver is able to decode its wanted message symbol by simply adding some index code symbols (broadcast symbols).

Given $K$ and $D$ the $K \times (D+1)$  matrix obtained by Algorithm I is called the $(K,D)$ AIR matrix and it is denoted by $\mathbf{L}_{K\times (D+1)}.$ The general form of the $(K,D)$ AIR matrix is shown in   Fig. \ref{fig1}. It consists of several submatrices (rectangular boxes) of different sizes as shown in Fig.\ref{fig1}. The location and sizes of these submatrices are used subsequently to prove the main results in the following section Theorems \ref{thm1} and \ref{thm2}.

The description of the submatrices are as follows: Let $m$ and $n$ be two positive integers and $n$ divides $m$. The following matrix  denoted by $\mathbf{I}_{m \times n}$ is a rectangular  matrix.
\begin{align}
\label{rcmatrix}
\mathbf{I}_{m \times n}=\left.\left[\begin{array}{*{20}c}
   \mathbf{I}_n  \\
   \mathbf{I}_n  \\
   \vdots  \\
   \mathbf{I}_n 
   \end{array}\right]\right\rbrace \frac{m}{n}~\text{number~of}~ \mathbf{I}_n~\text{matrices}
\end{align}
and $\mathbf{I}_{n \times m}$ is the transpose of $\mathbf{I}_{m \times n}.$ We will call the $\mathbf{I}_{m \times n}$ matrix the $(m \times n)$ identity matrix. 

		\begin{algorithm}
			{Algorithm to construct the AIR matrix $\mathbf{L}$ of size $K \times (D+1)$}
			\begin{algorithmic}[1]
				 \item []Given $K$ and $D$ let $\mathbf{L}=K \times (D+1)$ blank unfilled matrix.
				\item [Step 1]~~~
				\begin{itemize}
				\item[\footnotesize{1.1:}] Let $K=q(D+1)+r$ for $r < D+1$.
				\item[\footnotesize{1.2:}] Use $\mathbf{I}_{q(D+1) \times (D+1)}$ to fill the first $q(D+1)$ rows of the unfilled part of $\mathbf{L}$.
				\item[\footnotesize{1.3:}] If $r=0$,  Go to Step 3.
				\end{itemize}

				\item [Step 2]~~~
				\begin{itemize}
				\item[\footnotesize{2.1:}] Let $D+1=q^{\prime}r+r^{\prime}$ for $r^{\prime} < r$.
				\item[\footnotesize{2.2:}] Use $\mathbf{I}_{q^{\prime}r \times r}^{\mathsf{T}}$ to fill the first $q^{\prime}r$ columns of the unfilled part of $\mathbf{L}$.
			    \item[\footnotesize{2.3:}] If $r^{\prime}=0$,  go to Step 3.	
				\item[\footnotesize{2.4:}] $K\leftarrow r$ and $D+1\leftarrow r^{\prime}$.
				\item[\footnotesize{2.5:}] Go to Step 1.
				\end{itemize}
				\item [Step 3] Exit.
		
			\end{algorithmic}
			\label{algo1}
		\end{algorithm}

\begin{figure*}
\centering
\includegraphics[scale=0.58]{matrix213.eps}\\
~ $\mathbf{S}=\mathbf{I}_{\lambda_{l} \times \beta_l \lambda_{l}}$ if $l$ is even and ~$\mathbf{S}=\mathbf{I}_{\beta_l\lambda_{l} \times \lambda_{l}}$ otherwise.
\caption{AIR matrix of size $K \times (D+1)$.}
\label{fig1}
~ \\
\hrule
\end{figure*}

Towards explaining the other quantities shown in the AIR matrix shown in Fig. \ref{fig1}, for a given $K$  and $D,$ let  $\lambda_{-1}=D+1,\lambda_0=K-D-1$ and\begin{align}
\nonumber
D+1&=\beta_0 \lambda_0+\lambda_1, \nonumber \\
\lambda_0&=\beta_1\lambda_1+\lambda_2, \nonumber \\
\lambda_1&=\beta_2\lambda_2+\lambda_3, \nonumber \\
\lambda_2&=\beta_3\lambda_3+\lambda_4, \nonumber \\
&~~~~~~\vdots \nonumber \\
\lambda_i&=\beta_{i+1}\lambda_{i+1}+\lambda_{i+2}, \nonumber \\ 
&~~~~~~\vdots \nonumber \\ 
\lambda_{l-1}&=\beta_l\lambda_l.
\label{chain}
\end{align}
where $\lambda_{l+1}=0$ for some integer $l,$ $\lambda_i,\beta_i$ are positive integers and $\lambda_i < \lambda_{i-1}$ for $i=1,2,\ldots,l$. The number of submatrices in the AIR matrix is $l+2$ and the size of each submatrix is shown using $\lambda_i,\beta_i,$  $i \in [0:l].$ The submatrices are classified in to the following three types. 
\begin{itemize}
\item The first submatrix is the $I_{(D+1) \times (D+1)}$ matrix at the top of Fig. \ref{fig1} which is independent of $\lambda_i,\beta_i,$  $i \in [0:l].$ This will be referred as the $I_{(D+1)}$ matrix henceforth.
\item The set of matrices of the form $I_{\lambda_i \times \beta_i \lambda_i}$ for $i=0,2,4, \cdots$ (for all $i$ even) will be referred as the set of even-submatrices.
\item The set of matrices of the form $I_{\beta_i \lambda_i \times  \lambda_i}$ for $i=1,3,5, \cdots$ (for all $i$ odd) will be referred as the set of odd-submatrices. 
\end{itemize}
Note that the odd-submatrices are always "fat" and the even-submatrices are always "tall" including square matrices in both the sets. 
By the $i$-th submatrix is meant either an odd-submatrix or an even-submatrix for $ 0 \leq i \leq l.$  Also whenever  $\beta_0=0,$  the corresponding submatrix will not exist in the AIR matrix.  
To prove the main result in the following section the location of both the odd- and even-submatrices within the AIR matrix need to be identified. Towards this end, we define the following intervals. Let $R_0,R_1,R_2,\ldots,R_{\left\lfloor \frac{l}{2}\right\rfloor+1}$ be the intervals that will identify the rows of the submatrices  as given below:

\begin{itemize}
\item $R_0=[0:K-\lambda_0-1]$
\item $R_1=[K-\lambda_0:K-\lambda_2-1]$
\item $R_2=[K-\lambda_2:K-\lambda_4-1]$
\item []~~~~~~~~~~~~~~$\vdots$
\item $R_i=[K-\lambda_{2(i-1)}:K-\lambda_{2i}-1]$
\item []~~~~~~~~~~~~~~$\vdots$
\item $R_{\left\lfloor
\frac{l}{2}\right\rfloor}=[K-\lambda_{2(\left\lfloor
\frac{l}{2}\right\rfloor-1)}:K-\lambda_{2\left\lfloor
\frac{l}{2}\right\rfloor}-1]$
\item $R_{\left\lfloor \frac{l}{2}\right\rfloor+1}=[K-\lambda_{2\left\lfloor \frac{l}{2}\right\rfloor}:K-1]$,
\end{itemize} 
we have $R_0\cup R_1 \cup R_2 \cup \ldots \cup R_{\left\lfloor \frac{l}{2}\right\rfloor +1}=[0:K-1]$.

Let $C_0,C_1,\ldots,C_{\left\lceil \frac{l}{2}\right\rceil}$ be the intervals that will identify the columns of the submatrices  as given below:
\begin{itemize}

\item $C_0=[0:\beta_0\lambda_0-1]$ if $\beta_0 \geq 1$, else $C_0=\phi$
\item $C_1=[D-\lambda_1+1:D-\lambda_3]$
\item $C_2=[D-\lambda_3+1:D-\lambda_5]$
\item [] ~~~~~~~~~~~$\vdots$
\item $C_i=[D-\lambda_{2i-1}+1:D-\lambda_{2i+1}]$
\item []~~~~~~~~~~~~~~$\vdots$
\item $C_{\left\lceil
\frac{l}{2}\right\rceil-1}=[D-\lambda_{2\left\lceil
\frac{l}{2}\right\rceil-3}+1:D-\lambda_{2\left\lceil
\frac{l}{2}\right\rceil-1}]$
\item $C_{\left\lceil \frac{l}{2}\right\rceil}=[D-\lambda_{2\left\lceil \frac{l}{2}\right\rceil-1}+1:D]$
\end{itemize} 
we have $C_0\cup C_1 \cup C_2 \cup \ldots \cup C_{\left\lceil \frac{l}{2}\right\rceil}=[0:D]$.

Let $\mathbf{L}$ be the AIR matrix of size $K \times (D+1)$. In the matrix $\mathbf{L}$, the element $\mathbf{L}(j,k)$ is present in one of the submatrices: $\mathbf{I}_{D+1}$ or $\mathbf{I}_{\beta_{2i+1}\lambda_{2i+1} \times \lambda_{2i+1}}$ for $i \in [0:\lceil \frac{l}{2}\rceil-1]$ or $\mathbf{I}_{\lambda_{2i} \times \beta_{2i}\lambda_{2i}}$ for $i \in [0:\left\lfloor\frac{l}{2}\right\rfloor]$. Let $(j_R,k_R)$ be the (row-column) indices of $\mathbf{L}(j,k)$ within the submatrix in which $\mathbf{L}(j,k)$ is present. Then, for a given $\mathbf{L}(j,k)$, the indices $j_R$ and $k_R$ are as given below.
\begin{itemize}
\item If $\mathbf{L}(j,k)$ is present in $\mathbf{I}_{D+1}$, then $j_R=j$ and $k_R=k$.
\item If $\mathbf{L}(j,k)$ is present in $\mathbf{I}_{\lambda_{0} \times \beta_{0}\lambda_{0}}$, then   $j_R=j~\text{\textit{mod}}~(D+1)$ and $k_R=k$. 
\item If $L(j,k)$ is present in $\mathbf{I}_{\beta_{2i+1}\lambda_{2i+1} \times \lambda_{2i+1}}$ for $i \in [0:\lceil \frac{l}{2}\rceil-1]$, then \\ $j_R=j~\text{\textit{mod}}~(K-\lambda_{2i})$ and $k_R=k~\text{\textit{mod}}~(D+1-\lambda_{2i+1})$.
\item If $L(j,k)$ is present in $\mathbf{I}_{\lambda_{2i} \times \beta_{2i}\lambda_{2i}}$ for $i \in [1:\left\lfloor\frac{l}{2}\right\rfloor]$, then $j_R=j~\text{\textit{mod}}~(K-\lambda_{2i})$ and $k_R=k~\text{\textit{mod}}~(D+1-\lambda_{2i-1})$. 
\end{itemize}

In Definition \ref{def1} below we define several distances between the $1$s present in an AIR matrix. These distances are used to prove that AIR matrix can be used as optimal length encoding matrix for SNI-SUICP.  Figure \ref{sfig44} is useful to visualize the distances defined.

\begin{definition}
\label{def1}
Let $\mathbf{L}$ be the AIR matrix of size $K \times (D+1)$.
\begin{itemize}
\item [\textbf{(i)}] For $k\in [0:D]$ we have  $\mathbf{L}(k,k)=1.$  Let $k^{\prime}$ be the maximum integer such that $k^{\prime} > k$ and $\mathbf{L}(k^{\prime},k)=1$. Then $k^{\prime}-k,$ denoted by $d_{down}(k),$  is called the down-distance of $\mathbf{L}(k,k)$.  
\item [\textbf{(ii)}] Let $\mathbf{L}(j,k)=1$ and $j \geq D+1$. Let $j^{\prime}$ be the maximum integer such that $j^{\prime} < j$ and $\mathbf{L}(j^{\prime},k)=1$. Then $j-j^{\prime},$ denoted by $d_{up}(j,k),$ is called the up-distance of $\mathbf{L}(j,k).$ 

\item [\textbf{(iii)}]  Let $\mathbf{L}(j,k)=1$ and $\mathbf{L}(j,k)\in \mathbf{I}_{ \lambda_{2i} \times \beta_{2i} \lambda_{2i}}$ for $i \in [0:\lfloor \frac{l}{2}\rfloor]$. Let $k^{\prime}$ be the minimum integer such that $k^{\prime} > k$ and $\mathbf{L}(j,k^{\prime})=1$. Then $k^{\prime}-k,$ denoted by $d_{right}(j,k),$  is called the right-distance of $\mathbf{L}(j,k).$  

\item [\textbf{(iv)}] For $k\in [0:D-\lambda_l]$, let $d_{right}(k+d_{down}(k),k)=\mu_k.$  Let the number of $1$s in the $(k+\mu_k)$th column of $\mathbf{L}$ below $\mathbf{L}(k+d_{down}(k),k+\mu_k)$ be $p_k$ and these are  at a distance of $t_{k,1},t_{k,2},\ldots,t_{k,p_k}~(t_{k,1}<t_{k,2}<\ldots <t_{k,p_k})$ from $\mathbf{L}(k+d_{down}(k),k+\mu_k)$. Then, $t_{k,r}$ is called the $r$-th  down-distance of $\mathbf{L}(k+d_{down}(k),k+\mu_k)$, for $1 \leq r \leq p_k.$
\end{itemize}
\end{definition}

Notice that if $L(k+d_{down}(k),k+\mu_k) \in \mathbf{I}_{\lambda_{2i} \times   \beta_{2i} \lambda_{2i}}$ in $\mathbf{L}$ for $i \in [0:\lfloor \frac{l}{2}\rfloor]$, then $p_k=0$.

\begin{figure}
\centering
\includegraphics[scale=0.62]{proof22.eps}\\
\caption{Illustration of Definition \ref{def1}}
\label{sfig44}
\end{figure}

\begin{lemma}
\label{lemma1}
Let $k \in C_i$ for $i \in [0:\lceil\frac{l}{2}\rceil]$. Let $k~\text{mod}~(D+1-\lambda_{2i-1})=c\lambda_{2i}+d$ for some positive integers $c$ and $d$ $(d<\lambda_{2i})$. The down distance is given by
\begin{align}
\label{mdd}
d_{down}(k)=K-D-1+\lambda_{2i+1}+(\beta_{2i}-1-c)\lambda_{2i}.
\end{align}
\end{lemma}
\begin{proof}
Proof is given in Appendix A.
\end{proof}

\begin{lemma}
\label{lemma2}
The up-distance of $\mathbf{L}(j,k)$ is as given below.
\begin{itemize}
\item If $\mathbf{L}(j,k) \in \mathbf{I}_{\beta_{2i+1} \lambda_{2i+1} \times \lambda_{2i+1}}$  for $i \in [0:\lceil \frac{l}{2}\rceil-1]$, then $d_{up}(j,k)$ is $\lambda_{2i+1}$. 
\item If $\mathbf{L}(j,k) \in \mathbf{I}_{\lambda_{2i} \times  \beta_{2i} \lambda_{2i}}$  for $i\in [0:\lfloor \frac{l}{2}\rfloor]$ and $k_R=c\lambda_{2i}+d$ for some positive integer $c,d~(d < \lambda_{2i})$, then $d_{up}(j,k)$ is $\lambda_{2i-1}-c\lambda_{2i}$. 
\end{itemize}
\end{lemma}
\begin{proof}
Proof  is available in Appendix C of \cite{VaR3}.
\end{proof}

\begin{lemma}
\label{lemma3}
The right-distance of $\mathbf{L}(j,k)$ is as given below.
\begin{itemize}
\item If $k_R \in [0:(\beta_{2i}-1)\lambda_{2i}-1]$ for $i \in [0:\lfloor \frac{l}{2}\rfloor]$, then $d_{right}(j,k)$ is $\lambda_{2i}$. 
\item If $k_R \in [(\beta_{2i}-1)\lambda_{2i}:\beta_{2i} \lambda_{2i}-1]$ for $i \in [0:\lfloor \frac{l}{2}\rfloor-1]$, then $d_{right}(j,k)$ depends on $j_R$. If $j_R=c\lambda_{2i+1}+d$ for some positive integers $c,d~(d<\lambda_{2i+1})$, then  $d_{right}(j,k)$ is $\lambda_{2i}-c\lambda_{2i+1}$. 
\end{itemize}
\end{lemma}
\begin{proof}
Proof  is available in Appendix D of \cite{VaR3}.
\end{proof}

From Euclid algorithm and \eqref{chain}, we can write
\begin{align}
\label{chain2}
\lambda_l=\text{gcd}(K,D+1).
\end{align}
From \eqref{chain}, we have 
\begin{align}
\label{chain3}
\nonumber
&\lambda_0 > \lambda_1 >\ldots>\lambda_{2i}>\ldots>\lambda_l=\text{gcd}(K,D+1)~\text{and} \\& 
\lambda_{2i-1}-c\lambda_{2i} \geq \lambda_{2i+1} \geq \lambda_{l} = \text{gcd}(K,D+1)
\end{align}
for $i \in [0:\lfloor \frac{l}{2}\rfloor]$ and $c \leq \beta_{2i}$.

Define 
\begin{align*}
\tilde{C}_i=\{x+\lambda_0: \forall x \in C_i\}~\text{for}~i\in [0:\left \lceil\frac{l}{2}\right \rceil].
\end{align*}
That is, 
\begin{align}
\label{fact0}
\tilde{C}_i=[K-\lambda_{2i-1}:K-\lambda_{2i+1}-1].
\end{align}
 We have $C_0\cup C_1 \cup \ldots \cup C_{\left\lceil\frac{l}{2}\right\rceil}=[0:D-1]$, hence  
\begin{align*}
\tilde{C}_0 \cup \tilde{C}_1 \cup \ldots \cup \tilde{C}_{\left\lceil\frac{l}{2}\right\rceil}=[\lambda_0:K-1].
\end{align*}

\section{Optimal Index Coding for SNI-SUICP by using AIR matrices}
\label{sec3}
A scalar linear index code of length $D+1$ generated by an AIR matrix of size $K \times (D+1)$ is given by
\begin{align}
 [c_0~c_1~\ldots~c_{D}]=[x_0~x_1~\ldots~x_{K-1}]\mathbf{L}=\sum_{k=0}^{K-1}x_kL_k
\end{align}
where $L_k$ is the $k$th row of $\mathbf{L}$ for $k\in[0:K-1]$. We prove that for $k \in [0:K-1]$, every receiver $R_{k}$ decodes its wanted message $x_{k}$ by using $[c_0~c_1~\ldots~c_D]$ and its side-information.

In this section we show that the AIR matrix with parameter $K$ and $D+1$ is an encoding matrix for the optimal length code for our SNI-SUICP.

\begin{theorem}
\label{thm1}
Let $\mathbf{L}$ be the AIR matrix of size $K \times (D+1)$. The matrix $\mathbf{L}$ can be used as an encoding matrix for the SNI-SUICP with $K$ messages, $D$ interfering messages after and $U=\text{gcd}(K,D+1)-1$ interfering messages before the desired message.
\end{theorem}
\begin{proof}
Proof is given in Appendix B.
\end{proof}

\begin{theorem}
\label{thm2}
The capacity of SNI-SUICP with $K$ messages, $D$ interfering messages after and $U=\text{gcd}(K,D+1)-1$ interfering messages before the desired message is $\frac{1}{D+1}$.
\end{theorem}
\begin{proof}
In Theorem \ref{thm1}, we proved that AIR of size $K \times (D+1)$ can be used as an encoding matrix for this SNI-SUICP. The rate achieved by using AIR matrix is $\frac{1}{D+1}$. From \eqref{outerbound}, the rate of SNI-SUICP is always greater than or equal to $\frac{1}{D+1}$. Hence, the capacity of SNI-SUICP with $K$ messages, $D$ interfering messages after and $U=\text{gcd}(K,D+1)-1$ interfering messages before the desired message is $\frac{1}{D+1}$.
\end{proof}
\begin{remark}
Let $\tau_k$ be the set of broadcast symbols used by receiver $R_k$ to decode $x_k$. The number of broadcast symbols used by receiver $R_k$ by using AIRM as encoding matrix is given below:
\begin{itemize}
\item If $k \in [0:\lambda_0-1]$, then $|\tau_k|=1$.
\item If $k \in \tilde{D}_i$ for $i \in [0:\left\lceil\frac{l}{2}\right\rceil]$, then $|\tau_k|=2$
\item If $k \in \tilde{E}_i$ for $i \in [0:\left\lceil\frac{l}{2}\right\rceil-1]$, then $|\tau_k|=p_{k\prime}+2$, where $k^{\prime}=k-\lambda_0$ and $p_{k^\prime}$ is the number of $1$s below $\mathbf{L}(k+d_{down}(k^{\prime}),k^\prime+d_{right}(k^\prime+d_{down}(k^\prime),k^\prime))$ in AIR matrix.
\item If $k \in \tilde{E}_i$ for $i=\left\lceil\frac{l}{2}\right\rceil$, then $|\tau_k|=1$.
\end{itemize}
\end{remark}
\begin{remark}
Let $\gamma_k$ be the set of side-information used by receiver $R_k$ to decode $x_k$. Let $N_k$ be the number of message symbols present in $c_k$ for $k \in [0:D]$. The number of side-information used by receiver $R_k$ by using AIRM as encoding matrix is given below:
\begin{itemize}
\item If $k \in [0:\lambda_0-1]$, then $|\gamma_k|=N_{k~\text{mod}~(D+1)}-1$.
\item If $k \in \tilde{D}_i$ for $i \in [0:\left\lceil\frac{l}{2}\right\rceil]$, then $|\gamma_k|=N_{k^\prime}+N_{k^\prime+\mu_{k^\prime}}-3$, where $k^{\prime}=k-\lambda_0$.
\item If $k \in \tilde{E}_i$ for $i \in [0:\left\lceil\frac{l}{2}\right\rceil-1]$, then $|\gamma_k|=N_{k^\prime}+N_{k^\prime+\mu_{k^\prime}}+\sum_{j=1}^{p_{k^\prime}} N_{k^\prime+t_{k^\prime,j}}-2p_{k^\prime}-3$.
\item If $k \in \tilde{E}_i$ for $i=\left\lceil\frac{l}{2}\right\rceil$, then  $|\gamma_k|=N_{k^\prime}-1$.
\end{itemize}
\end{remark}

\begin{example}
\label{ex1}
Consider a SNI-SUICP with $K=12,D=7,U=3$. The capacity of this SNI-SUICP is $\frac{1}{8}$. AIRM of size $12 \times 8$ can be used as an optimal length encoding matrix for this SNI-SUICP. The encoding matrix $\mathbf{L}_{12 \times 8}$ is given below. The code symbols and side-information used by each receiver to decode its wanted message is given in Table \ref{table1}.

\arraycolsep=0.5pt
\setlength\extrarowheight{-3.0pt}
{
$$\mathbf{L}_{12 \times 8}=\left[
\begin{array}{cccccccccc}
1 & 0 & 0 & 0 & 0 & 0 & 0 & 0\\
0 & 1 & 0 & 0 & 0 & 0 & 0 & 0\\
0 & 0 & 1 & 0 & 0 & 0 & 0 & 0\\
0 & 0 & 0 & 1 & 0 & 0 & 0 & 0\\
0 & 0 & 0 & 0 & 1 & 0 & 0 & 0\\
0 & 0 & 0 & 0 & 0 & 1 & 0 & 0\\
0 & 0 & 0 & 0 & 0 & 0 & 1 & 0\\
0 & 0 & 0 & 0 & 0 & 0 & 0 & 1\\
1 & 0 & 0 & 0 & 1 & 0 & 0 & 0\\
0 & 1 & 0 & 0 & 0 & 1 & 0 & 0\\
0 & 0 & 1 & 0 & 0 & 0 & 1 & 0\\
0 & 0 & 0 & 1 & 0 & 0 & 0 & 1\\
 \end{array}
\right]$$
}

\begin{table}[h]
\centering
\setlength\extrarowheight{0pt}
\begin{tabular}{|c|c|c|c|c|c|c|c|}
\hline
\textbf{$R_k$} &$\mathcal{W}_k$&$D_{max}(k)$&$\mu_k$&$\mu_{k^{\prime}}$&$t_{k^{\prime},1}$&$\tau_k$& $\gamma_k$ \\
\hline
$R_0$ & $x_0$ & 8&4&-&-&$c_0$&$x_8$ \\
\hline
$R_1$ & $x_1$ & 8&4&-&-&$c_1$&$x_{9}$ \\
\hline
$R_2$ & $x_2$ & 8&4&-&-&$c_2$&$x_{10}$ \\
\hline
$R_3$ & $x_3$ &8&4&-&-&$c_3$&$x_{11}$ \\
\hline
$R_4$ & $x_4$ & 4&-&4&-&$c_0,c_4$&$x_{0}$ \\
\hline
$R_5$ & $x_5$ & 4&-&4&-&$c_1,c_5$&$x_{1}$ \\
\hline
$R_6$ & $x_6$ & 4&-&4&-&$c_2,c_6$&$x_{2}$ \\
\hline
$R_7$ & $x_7$ & 4&-&4&-&$c_3,c_7$&$x_{3}$ \\
\hline
$R_8$ & $x_8$ &-&-&-&-&$c_8$&$x_{4}$ \\
\hline
$R_{9}$ & $x_{9}$ &-&-&-&-&$c_9$&$x_{5}$ \\
\hline
$R_{10}$ & $x_{10}$ & -&-&-&-&$c_{10}$&$x_{6}$ \\
\hline
$R_{11}$ & $x_{11}$ &-&-&-&-&$c_{11}$&$x_{7}$ \\
\hline
\end{tabular}
\vspace{5pt}
\caption{Decoding of SNI-SUICP given in Example \ref{ex1}. In this example $k^\prime=k-\lambda_0=4$.}
\label{table1}
\vspace{-5pt}
\end{table}
\end{example}
\begin{example}
\label{ex4}
Consider a SNI-SUICP with $K=33,D=20,U=2$. The capacity of this SNI-SUICP is $\frac{1}{21}$. AIRM of size $33 \times 21$ can be used as an optimal length encoding matrix for this SNI-SUICP. For this SNI-SUICP, $D+1=21, \lambda_1=9,\lambda_2=3,\beta_0=1,\beta_1=1,\beta_2=3,$ and $l=2$. The encoding matrix for this SNI-SUICP is shown in Fig. \ref{ex4matrix}. The code symbols and side-information used by each receiver to decode its wanted message is given in Table \ref{table3}.
\end{example}
\begin{figure*}[t]
\arraycolsep=0.9pt
\begin{small}
{
$$\mathbf{L}_{33 \times 21}=\left[
\begin{array}{ccccccccccccccccccccccccccc}
1 & 0 & 0 & 0 & 0 & 0 & 0 & 0 & 0 & 0 & 0 & 0& 0 & 0 & 0 & 0 & 0 & 0 & 0 & 0 & 0 \\
0 & 1 & 0 & 0 & 0 & 0 & 0 & 0 & 0 & 0 & 0 & 0& 0 & 0 & 0 & 0 & 0 & 0 & 0 & 0 & 0 \\
0 & 0 & 1 & 0 & 0 & 0 & 0 & 0 & 0 & 0 & 0 & 0& 0 & 0 & 0 & 0 & 0 & 0 & 0 & 0 & 0 \\
0 & 0 & 0 & 1 & 0 & 0 & 0 & 0 & 0 & 0 & 0 & 0& 0 & 0 & 0 & 0 & 0 & 0 & 0 & 0 & 0 \\
0 & 0 & 0 & 0 & 1 & 0 & 0 & 0 & 0 & 0 & 0 & 0& 0 & 0 & 0 & 0 & 0 & 0 & 0 & 0 & 0 \\
0 & 0 & 0 & 0 & 0 & 1 & 0 & 0 & 0 & 0 & 0 & 0& 0 & 0 & 0 & 0 & 0 & 0 & 0 & 0 & 0 \\
0 & 0 & 0 & 0 & 0 & 0 & 1 & 0 & 0 & 0 & 0& 0 & 0 & 0 & 0 & 0 & 0 & 0 & 0 & 0 & 0 \\
0 & 0 & 0 & 0 & 0 & 0 & 0 & 1 & 0 & 0 & 0 & 0& 0 & 0 & 0 & 0 & 0 & 0 & 0 & 0 & 0 \\
0 & 0 & 0 & 0 & 0 & 0 & 0 & 0 & 1 & 0 & 0 & 0& 0 & 0 & 0 & 0 & 0 & 0 & 0 & 0 & 0 \\
0 & 0 & 0 & 0 & 0 & 0 & 0 & 0 & 0 & 1 & 0 & 0& 0 & 0 & 0 & 0 & 0 & 0 & 0 & 0 & 0 \\
0 & 0 & 0 & 0 & 0 & 0 & 0 & 0 & 0 & 0 & 1 & 0& 0 & 0 & 0 & 0 & 0 & 0 & 0 & 0 & 0 \\
0 & 0 & 0 & 0 & 0 & 0 & 0 & 0 & 0 & 0 & 0 & 1& 0 & 0 & 0 & 0 & 0 & 0 & 0 & 0 & 0 \\
0 & 0 & 0 & 0 & 0 & 0 & 0 & 0 & 0 & 0 & 0 & 0& 1 & 0 & 0 & 0 & 0 & 0 & 0 & 0 & 0 \\
0 & 0 & 0 & 0 & 0 & 0 & 0 & 0 & 0 & 0 & 0 & 0& 0 & 1 & 0 & 0 & 0 & 0 & 0 & 0 & 0 \\
0 & 0 & 0 & 0 & 0 & 0 & 0 & 0 & 0 & 0 & 0 & 0& 0 & 0 & 1 & 0 & 0 & 0 & 0 & 0 & 0 \\
0 & 0 & 0 & 0 & 0 & 0 & 0 & 0 & 0 & 0 & 0 & 0& 0 & 0 & 0 & 1 & 0 & 0 & 0 & 0 & 0 \\
0 & 0 & 0 & 0 & 0 & 0 & 0 & 0 & 0 & 0 & 0 & 0& 0 & 0 & 0 & 0 & 1 & 0 & 0 & 0 & 0 \\
0 & 0 & 0 & 0 & 0 & 0 & 0 & 0 & 0 & 0 & 0 & 0& 0 & 0 & 0 & 0 & 0 & 1 & 0 & 0 & 0 \\
0 & 0 & 0 & 0 & 0 & 0 & 0 & 0 & 0 & 0 & 0 & 0& 0 & 0 & 0 & 0 & 0 & 0 & 1 & 0 & 0 \\
0 & 0 & 0 & 0 & 0 & 0 & 0 & 0 & 0 & 0 & 0 & 0& 0 & 0 & 0 & 0 & 0 & 0 & 0 & 1 & 0 \\
0 & 0 & 0 & 0 & 0 & 0 & 0 & 0 & 0& 0  & 0 & 0 & 0& 0 & 0 & 0 & 0 & 0 & 0 & 0 & 1 \\
\hline
1 & 0 & 0 & 0 & 0 & 0 & 0 & 0 & 0 & 0 & 0 &0&\vline 1 & 0& 0 & 0 & 0 & 0 & 0 & 0 & 0 \\
0 & 1 & 0 & 0 & 0 & 0 & 0 & 0 & 0 & 0 & 0 &0&\vline 0& 1 & 0 & 0 & 0 & 0 & 0 & 0 & 0 \\
0 & 0 & 1 & 0 & 0 & 0 & 0 & 0 & 0 & 0 & 0 &0&\vline 0& 0 & 1 & 0 & 0 & 0 & 0 & 0 & 0 \\
0 & 0 & 0 & 1 & 0 & 0 & 0 & 0 & 0 & 0 & 0 &0&\vline 0& 0 & 0 & 1 & 0 & 0 & 0 & 0 & 0 \\
0 & 0 & 0 & 0 & 1 & 0 & 0 & 0 & 0 & 0 & 0 &0&\vline 0& 0 & 0 & 0 & 1 & 0 & 0 & 0 & 0 \\
0 & 0 & 0 & 0 & 0 & 1 & 0 & 0 & 0 & 0 & 0 &0&\vline 0& 0 & 0 & 0 & 0 & 1 & 0 & 0 & 0 \\
0 & 0 & 0 & 0 & 0 & 0 & 1 & 0 & 0 & 0 & 0 &0&\vline 0& 0 & 0 & 0 & 0 & 0 &  1 & 0 & 0 \\
0 & 0 & 0 & 0 & 0 & 0 & 0 & 1 & 0 & 0 & 0 &0&\vline 0& 0 & 0 & 0 & 0 & 0 &  0 & 1 & 0 \\
0 & 0 & 0 & 0 & 0 & 0 & 0 & 0 & 1 & 0 & 0 &0&\vline 0& 0 & 0 & 0 & 0 & 0 & 0 & 0 &  1\\
0 & 0 & 0 & 0 & 0 & 0 & 0 & 0 & 0 & 1 & 0 &0&\vline \overline{1}& \overline{0} & \overline{0} & \overline{1} & \overline{0} & \overline{0} & \overline{1} & \overline{0} & \overline{0}  \\
0 & 0 & 0 & 0 & 0 & 0 & 0 & 0 & 0 & 0 & 1 &0&\vline 0& 1 & 0 & 0 & 1 & 0 & 0 & 1 &0  \\
0 & 0 & 0 & 0 & 0 & 0 & 0 & 0 & 0 & 0 & 0 &1&\vline 0& 0 & 1 & 0 & 0 & 1 & 0 & 0 &1  \\
 \end{array}
\right]$$
\caption{Encoding matrix for the SNC-SUICP in Example \ref{ex4}.}
\label{ex4matrix}
}
\end{small}
\end{figure*}
\begin{table*}
\centering
\setlength\extrarowheight{0pt}
\begin{tabular}{|c|c|c|c|c|c|c|c|c|}
\hline
\textbf{$R_k$} &$\mathcal{W}_k$&$D_{max}(k)$&$\mu_k$&$\mu_{k^{\prime}}$&$t_{k,1}$&$t_{k^{\prime},1}$& $\tau_k$&$\gamma_k$ \\
\hline
$R_0$ & $x_0$ & 21&12&~&9&~&$c_0$&$x_{21}$ \\
\hline
$R_1$ & $x_1$ & 21&12&~&9&~&$c_1$&$x_{22}$ \\
\hline
$R_2$ & $x_2$ & 21&12&~&9&~&$c_2$&$x_{23}$ \\
\hline
$R_3$ & $x_3$ & 21&12&~&6&~&$c_3$&$x_{24}$ \\
\hline
$R_4$ & $x_4$ & 21&12&~&6&~&$c_4$&$x_{25}$ \\
\hline
$R_5$ & $x_5$ & 21&12&~&6&~&$c_5$&$x_{26}$ \\
\hline
$R_6$ & $x_6$ & 21&12&~&3&~&$c_6$&$x_{27}$ \\
\hline
$R_7$ & $x_7$ & 21&12&~&3&~&$c_7$&$x_{28}$ \\
\hline
$R_8$ & $x_8$ &21&12&~&3&~&$c_8$&$x_{29}$ \\
\hline
$R_{9}$ & $x_{9}$ & 21&3&~&-&~&$c_9$&$x_{30}$ \\
\hline
$R_{10}$ & $x_{10}$ & 21&3&~&-&~&$c_{10}$&$x_{31}$ \\
\hline
$R_{11}$ & $x_{11}$ & 21&3&~&-&~&$c_{11}$&$x_{32}$ \\
\hline
$R_{12}$ & $x_{12}$ & 18&3&12&-&9&$c_{0},c_{9},c_{12}$&$x_{0},x_{9}$ \\
\hline
$R_{13}$ & $x_{13}$ & 18&3&12&-&9&$c_{1},c_{10},c_{13}$&$x_{1},x_{10}$ \\
\hline
$R_{14}$ & $x_{14}$ & 18&3&12&-&9&$c_{2},c_{11},c_{14}$&$x_{2},x_{11}$ \\
\hline
$R_{15}$ & $x_{15}$ & 15&3&12&-&6&$c_{3},c_{12},c_{15}$&$x_{3},x_{9}$ \\
\hline
$R_{16}$ & $x_{16}$ & 15&3&12&-&6&$c_{4},c_{13},c_{16}$&$x_{4},x_{10}$ \\
\hline
$R_{17}$ & $x_{17}$ & 15&3&12&-&6&$c_{5},c_{14},c_{17}$&$x_{5},x_{11}$ \\
\hline
$R_{18}$ & $x_{18}$ & 12&3&12&-&3&$c_{6},c_{15},c_{18}$&$x_{6},x_{9}$ \\
\hline
$R_{19}$ & $x_{19}$ & 12&3&12&-&3&$c_{7},c_{16},c_{19}$&$x_{7},x_{10}$ \\
\hline
$R_{20}$ & $x_{20}$ & 12&3&12&-&3&$c_{8},c_{17},c_{20}$&$x_{8},x_{11}$ \\
\hline
$R_{21}$ & $x_{21}$ &-&-&3&-&~&$c_{9},c_{12}$&$x_9,x_{12}$ \\
\hline
$R_{22}$ & $x_{22}$ &-&-&3&-&~&$c_{10},c_{13}$&$x_{10},x_{13}$ \\
\hline
$R_{23}$ & $x_{23}$ &-&-&3&-&~&$c_{11},c_{14}$&$x_{11},x_{14}$ \\
\hline
$R_{24}$ & $x_{24}$ &-&-&3&-&~&$c_{12},c_{15}$&$x_{12},x_{15},x_{21}$ \\
\hline
$R_{25}$ & $x_{25}$ &-&-&3&-&~&$c_{13},c_{16}$&$x_{13},x_{16},x_{22}$ \\
\hline
$R_{26}$ & $x_{26}$ &-&-&3&-&~&$c_{14},c_{17}$&$x_{14},x_{17},x_{23}$ \\
\hline
$R_{27}$ & $x_{27}$ &-&-&3&-&~&$c_{15},c_{18}$&$x_{15},x_{18},x_{24}$ \\
\hline
$R_{28}$ & $x_{28}$ &-&-&3&-&~&$c_{16},c_{19}$&$x_{16},x_{19},x_{25}$ \\
\hline
$R_{29}$ & $x_{29}$ &-&-&3&-&~&$c_{17},c_{20}$&$x_{17},x_{20},x_{26}$ \\
\hline
$R_{30}$ & $x_{30}$ &-&-&3&-&~&$c_{18}$&$x_{18},x_{27}$ \\
\hline
$R_{31}$ & $x_{31}$ &-&-&3&-&~&$c_{19}$&$x_{19},x_{28}$ \\
\hline
$R_{32}$ & $x_{32}$ &-&-&3&-&~&$c_{20}$&$x_{20},x_{29}$ \\
\hline
\end{tabular}
\vspace{5pt}
\caption{Decoding of SNI-SUICP given in Example \ref{ex4}. In this example $k^\prime=k-\lambda_0=12$.}
\label{table3}
\vspace{-5pt}
\end{table*}
\begin{example}
\label{ex5}
Consider a SNI-SUICP with $K=432,D=175,U=15$. The capacity of this SNI-SUICP is $\frac{1}{176}$. For this SNI-SUICP, we have $D+1=176, \lambda_1=176,\lambda_2=80,\lambda_3=16,\beta=0,\beta_1=1,\beta_2=2,\beta_3=5$ and $l=3$. AIRM of size $432 \times 176$ given in Fig. \ref{sfig5} can be used as an optimal length encoding matrix for this SNI-SUICP. 
\begin{figure}
\centering
\includegraphics[scale=0.4]{matrix12.eps}\\
\caption{AIRM of size $432 \times 176$.}
\label{sfig5}
\end{figure}
\end{example}
\begin{example}
\label{ex6}
Consider a SNI-SUICP with $K=432,D=255,U=15$. The capacity of this SNI-SUICP according to Theorem \ref{thm2} is $\frac{1}{256}$. For this SNI-SUICP, we have $D+1=256, \lambda_1=80,\lambda_2=16,\beta_0=1,\beta_1=2,\beta_2=5$ and $l=2$. AIRM of size $432 \times 256$ given in Fig. \ref{sfig6} can be used as an optimal length encoding matrix for this SNI-SUICP. 
\begin{figure}
\centering
\includegraphics[scale=0.4]{matrix121.eps}\\
\caption{AIRM of size $432 \times 256$.}
\label{sfig6}
\end{figure}
\end{example}

\section{Discussion}
\label{sec4}
In this paper, we derived the capacity of SNI-SUICP and proposed optimal length coding scheme to achieve the capacity. Some of the interesting directions of further research
are as follows:  
\begin{itemize}
\item The capacity and optimal coding for SNI-SUICP with arbitrary $U$ and $D$ is a challenging open problem.
\item Maleki \textit{et al.} \cite{MCJ} proved the capacity of X network setting with local connectivity and $ML$ number of messages when the number of source-receiver pairs ($M$) tends to infinity. The capacity of this network is $\frac{2}{L(L+1)}$ per message. However, for finite $M$, the capacity is upper bounded by $\frac{2}{L(L+1)}$ but unknown.
\end{itemize}

\section*{APPENDIX A}
\subsection*{Proof of Lemma \ref{lemma1} }
\textit{Case (i)}: $l$ is even and $k \in C_i$ for $i \in [0:\lceil\frac{l}{2}\rceil]$ or  $l$ is odd and $k \in C_i$ for $i \in [0:\lceil\frac{l}{2}\rceil-1]$.

In this case, from the definition of down distance, we have $L(k+d_{down}(k),k) \in \mathbf{I}_{ \lambda_{2i} \times \beta_{2i}\lambda_{2i}}$.
\begin{figure*}
\centering
\includegraphics[scale=0.67]{matrix.eps}\\
\caption{Maximum-down distance calculation}
\label{afig1}
\end{figure*}
Let $k~\text{mod}~(D+1-\lambda_{2i-1})=c\lambda_{2i}+d$ for some positive integers $c$ and $d$ $(d<\lambda_{2i})$. 
From Figure \ref{afig1}, we have 
\begin{align}
\label{aeq1}
d_{down}(k)=d_1+d_2+d_3,
\end{align}
and 
\begin{align}
\label{aeq2}
\nonumber
&d_1=(D+1)-k, \\&
\nonumber
d_2=K-D-1-\lambda_{2i}, \\&
d_3=k-(D+1-\lambda_{2i-1})-c\lambda_{2i}.
\end{align}

By using \eqref{aeq1} and \eqref{aeq2}, we have 
\begin{align}
\label{aeq3}
\nonumber
d_{down}(k)&=d_1+d_2+d_3\\&
\nonumber
=(D+1)-k+\underbrace{K-D-1-\lambda_{2i}}_{d_2}\\&+\underbrace{k-(D+1-\lambda_{2i-1})-c\lambda_{2i}}_{d_3}\\&
\nonumber
=K-D-1+\lambda_{2i-1}-(c+1)\lambda_{2i}.
\end{align}

By replacing $\lambda_{2i-1}$ with $\beta_{2i}\lambda_{2i}+\lambda_{2i+1}$ in \eqref{aeq3}, we get 
\begin{align*}
d_{down}(k)=K-D-1+\lambda_{2i+1}+(\beta_{2i}-1-c)\lambda_{2i}.
\end{align*}
\textit{Case (ii)}: $l$ is odd and $k \in C_{\lceil\frac{l}{2}\rceil}$.

In this case, from the definition of down distance, we have $L(k+d_{down}(k),k) \in \mathbf{I}_{\beta_{l}\lambda_{l} \times \lambda_{l}}$.
\begin{figure*}
\centering
\includegraphics[scale=0.67]{amatrix.eps}\\
\caption{Maximum-down distance calculation}
\label{afig2}
\end{figure*}
From Figure \ref{afig2}, we have 
\begin{align}
\label{aeq32}
d_{down}(k)=d_1+d_2+d_3,
\end{align}
and
\begin{align}
\label{aeq4}
\nonumber
&d_1=(D+1)-k, \\&
\nonumber
d_2=K-D-1-\beta_l\lambda_{l}, \\&
d_3=\beta_l\lambda_l-d_5.
\end{align}

We have $L(k,k) \in \mathbf{I}_{D+1}$ and $L(k+d_{down}(k),k) \in \mathbf{I}_{\lambda_l}$ of $\mathbf{I}_{\beta_{l}\lambda_{l} \times \lambda_{l}}$ as shown in Figure \ref{afig2}. Hence, we have $d_1=d_4$ and $d_4=d_5$.
By using \eqref{aeq32} and \eqref{aeq4}, we have 
\begin{align*}
\nonumber
d_{down}(k)&=d_1+d_2+d_3\\&
\nonumber
=d_1+\underbrace{K-D-1-\beta_l\lambda_{l}}_{d_2}+\underbrace{\beta_l\lambda_l-d_1}_{d_3}=K-D-1.
\end{align*}

For $i=\left \lceil \frac{l}{2} \right \rceil$, we have $\lambda_{2\left \lceil \frac{i}{2} \right \rceil}=\lambda_{2\left \lceil \frac{i}{2} \right \rceil+1}=0$. We can write $K-D-1$ as $K-D-1+\lambda_{2\left \lceil \frac{l}{2} \right \rceil+1}+(\beta_{2\left\lceil\frac{l}{2}\right\rceil}-1-c)\lambda_{2\left \lceil \frac{l}{2} \right \rceil}$. Hence 
\begin{align*}
d_{down}(k)=K-D-1+\lambda_{2i+1}+(\beta_{2i}-1-c)\lambda_{2i}.
\end{align*}


\section*{APPENDIX B}
It turns out that the interval $\tilde{C}_{i}$ defined in \eqref{fact0} for $i \in [0:\left\lceil\frac{l}{2}\right\rceil]$
needs to be partitioned into two  as $\tilde{C}_{i} =\tilde{D}_{i}\cup \tilde{E}_{i}$ as given below to prove the main result Theorem \ref{thm1}. Let 
\begin{align}
\label{fact12}
&\tilde{D}_i=[K-\lambda_{2i-1}:K-\lambda_{2i-1}+(\beta_{2i}-1)\lambda_{2i}-1] \\&
\tilde{E}_i=[K-\lambda_{2i-1}+(\beta_{2i}-1)\lambda_{2i}:K-\lambda_{2i+1}-1]. 
\end{align}
for $i \in [0:\left\lceil\frac{l}{2}\right\rceil]$.

\subsection*{Proof of Theorem \ref{thm1}}
A scalar linear index code of length $D+1$ generated by an AIR matrix of size $K \times (D+1)$ is given by
\begin{align}
\label{code1}
 [c_0~c_1~\ldots~c_{D}]=[x_0~x_1~\ldots~x_{K-1}]\mathbf{L}=\sum_{k=0}^{K-1}x_kL_k
\end{align}
where $L_k$ is the $k$th row of $\mathbf{L}$ for $k\in[0:K-1]$. We prove that for $k \in [0:K-1]$, every receiver $R_{k}$ decodes its wanted message $x_{k}$ by using $[c_0~c_1~\ldots~c_D]$ and its side-information. 

{\bf {Case (i)}:} $k \in [0:\lambda_0-1]$

If $K-D-1 < \left\lceil\frac{K}{2}\right \rceil$, the broadcast symbol $c_{k}$ is given by $c_{k}=x_{k}+x_{k+D+1}$. In $c_{k}$, the message symbol $x_{k+D+1}$ is in the side-information of receiver $R_{k}$. Hence, $R_{k}$ can decode its wanted message symbol $x_{k}$ from $c_{k}$. 

If $K-D-1 \geq \left \lceil \frac{K}{2}\right \rceil$, we show that $R_k$ can decode $x_k$ from $c_{k~\text{mod}~(D+1)}$. In this case, from \eqref{chain}, we have $\beta_0=0$ and $\lambda_1=D+1$. 

If $k \leq D$ ($k~\text{mod}~(D+1)=k$), From Lemma \ref{lemma2}, we have 
\begin{align}
\label{fact52}
&d_{up}(k+D+1,k)=D+1.
\end{align}

This indicates that $x_{k+1},x_{k+2},\ldots,x_{k+D}$ are not present in $c_{k}$. From Lemma \ref{lemma1}, we have 
\begin{align}
\label{fact521}
\nonumber
d_{down}(k)&=K-D-1+\lambda_{2i+1}+(\beta_{2i}-1-c)\lambda_{2i}\\&
 \leq K-\lambda_l=K-\text{gcd}(K,D+1).
\end{align}

This indicates that $x_{k-\text{gcd}(K,D+1)+1},\ldots,x_{k-1}$ are not present in $c_{k}$. Hence, every message symbol in $c_{k}$ is in the side-information of $R_k$ excluding the message symbol $x_k$ and $R_k$ can decode $x_k$.

If $k \in [D+1:\lambda_0-1]$, From Lemma \ref{lemma2}, we have 
\begin{align*}
d_{up}(k,k~\text{mod}~(D+1))&=d_{up}(k+D+1,k~\text{mod}~(D+1))\\&=D+1.
\end{align*}

Hence, $c_{k~\text{mod}~(D+1)}$ does not contain message symbols from the set $\{x_{k-D},\ldots,x_{k-1}\}\cup \{x_{k+1},\ldots,x_{k+D}\}$ and $R_k$ can decode $x_k$ from $c_{k~\text{mod}~(D+1)}$.


{\bf Case (ii):} $k \in \tilde{D_{i}}$  for $i\in [0:\left \lceil \frac{l}{2}\right \rceil]$.

Let $k^{\prime}=k-\lambda_0$. In this case, we have $k^{\prime}_R \in [0:(\beta_{2i}-1)\lambda_{2i}-1]$  for $i\in [0:\left \lceil\frac{l}{2}\right \rceil]$. Let $k^{\prime}_R=c\lambda_{2i}+d$ for some positive integers $c$ and $d$ and $d<\lambda_{2i}$. From Lemma \ref{lemma3}, we have $\mu_{k^{\prime}}=\lambda_{2i}$, from Definition \ref{def1}, we have  $t_{k^{\prime},r}=0$ for $r \in [1:p_k]$. From Lemma \ref{lemma1}, we have 
\begin{align}
\label{fact1}
\nonumber
d_{down}(k^{\prime})&=K-D-1+\lambda_{2i+1}+(\beta_{2i}-1-c)\lambda_{2i}\\&
=K-D-1+\lambda_{2i-1}-(c+1)\lambda_{2i}.
\end{align}
From Lemma \ref{lemma2}, we have 
\begin{align}
\label{fact2}
\nonumber
&d_{up}(k^{\prime}+d_{down}(k^{\prime}),k^{\prime})=\lambda_{2i-1}-c\lambda_{2i}\\& d_{up}(k^{\prime}+d_{down}(k^{\prime}),k^{\prime}+\mu_{k^{\prime}})=\lambda_{2i-1}-(c+1)\lambda_{2i}.
\end{align}
From \eqref{fact1} and \eqref{fact2}
\begin{align}
\label{fact3}
\nonumber
d_{down}(k^{\prime})-d_{up}(k^{\prime}&+d_{down}(k^{\prime}),k^{\prime}+\mu_{k^{\prime}})\\&=K-D-1.
\end{align}

This indicates that $x_{k}$ is present in the code symbol $c_{k^{\prime}+\mu_{k^{\prime}}}$ and among $D$ interfering messages after $x_k$ ($x_{k+1},x_{k+2},\ldots,x_{k+D}$), only  $x_{k^{\prime}+d_{down}(k^{\prime})}$ is present in $c_{k^{\prime}+\mu_{k^{\prime}}}$. Fig. \ref{sfig31} and \ref{sfig3} illustrate this.
From \eqref{fact2}, 
\begin{align}
\label{fact4}
\nonumber
d_{up}(k^{\prime}+d_{down}(k^{\prime}),k^{\prime})-d_{up}(k^{\prime}&+d_{down}(k^{\prime}),k^{\prime}+\mu_{k^{\prime}})\\&=\lambda_{2i} \geq gcd(K,D+1).
\end{align}

This along with \eqref{fact3} indicates that every message symbol in $c_{k^{\prime}}$ is in the side-information of $R_{k}$ except $x_{k^{\prime}+d_{down}(k^{\prime})}$. Fig. \ref{sfig31} and \ref{sfig3} illustrate this. Hence, every message symbol in $c_{k^{\prime}}+c_{k^{\prime}+\mu_{k^{\prime}}}$ is in the side-information of $R_{k}$ and $R_{k}$ decodes $x_{k}$. 
\begin{figure*}
\centering
\includegraphics[scale=0.50]{proof5.eps}\\
\caption{Decoding for $k=k^{\prime}+\lambda_0 \in \tilde{C}_i$.}
\label{sfig31}
\end{figure*}

\begin{figure}
\centering
\includegraphics[scale=0.70]{proof8.eps}\\
\caption{Decoding for $k=k^{\prime}+\lambda_0 \in \tilde{D}_i$.}
\label{sfig3}
\end{figure}

{\bf Case (iii):} $k \in \tilde{E_{i}}$  for $i\in [0:\lceil\frac{l}{2}\rceil-1]$.

Let $k^{\prime}=k-\lambda_0$. In this case, we have $k^{\prime}_R \in [(\beta_{2i}-1)\lambda_{2i}:\beta_{2i}\lambda_{2i}-1]$ for $i\in [0:\lceil\frac{l}{2}\rceil]$. Let $k^{\prime}_R=(\beta_{2i}-1)\lambda_{2i}+c\lambda_{2i+1}+d$ for some positive integers $c,d \ (d<\lambda_{2i+1})$. We have $k^{\prime}=D+1-\lambda_{2i-1}+k^{\prime}_R$. From  Lemma \ref{lemma1}, we have 
\begin{align}
\label{fact5}
d_{down}(k^{\prime})=K-D-1+\lambda_{2i+1}. 
\end{align}
From Lemma \ref{lemma3}, we have 
\begin{align}
\label{fact6}
\nonumber
\mu_{k^{\prime}}&=d_{right}(k^{\prime}+d_{down}(k^{\prime}),k^{\prime})\\&
\nonumber
=d_{right}(D+1-\lambda_{2i-1}+k^{\prime}_R\\&
\nonumber 
~~~~~~~~~~~~+K-D-1+\lambda_{2i+1},k^{\prime})\\&
\nonumber
=d_{right}(K-\lambda_{2i}+c\lambda_{2i+1}+d,k^{\prime})\\&
=\lambda_{2i}-c\lambda_{2i+1}.
\end{align}
From Lemma \ref{lemma2}, we have
\begin{align}
\label{fact7}
d_{up}(k^{\prime}+d_{down}(k^{\prime}),k^{\prime}+\mu_{k^{\prime}})=\lambda_{2i+1}.
\end{align}
From \eqref{fact1} and \eqref{fact2}
\begin{align}
\label{fact8}
\nonumber
d_{down}(k^{\prime})-d_{up}(k^{\prime}&+d_{down}(k^{\prime}),k^{\prime}+\mu_{k^{\prime}})\\&=K-D-1.
\end{align}

This indicates that $x_{k}$ is present in the code symbol $c_{k^{\prime}+\mu_{k^{\prime}}}$ and among $D$ interfering messages after $x_k$ ($x_{k+1},x_{k+2},\ldots,x_{k+D}$), the interfering messages $x_{k^{\prime}+d_{down}(k^{\prime})}$ and $x_{k^{\prime}+t_{k^{\prime},r}+d_{down}(k^{\prime})}$ for $r \in [1:p_{k^{\prime}}]$ are present in $c_{k^{\prime}+\mu_{k^{\prime}}}$. Fig. \ref{sfig4} is useful to understand this.

From Lemma \ref{lemma1} and Definition \ref{def1}, $k^{\prime}+d_{down}(k^{\prime})+t_{k^{\prime},p_{k^{\prime}}}$ is always less than the number of rows in the matrix $\mathbf{L}$. That is, $k^{\prime}+t_{k^{\prime},p}+d_{down}(k^{\prime})<K$. Hence, we have 
\begin{align}
\label{fact9}
\nonumber
t_{k^{\prime},p}&< K-k^{\prime}-d_{down}(k^{\prime})\\&
\nonumber
=K-(D+1-\lambda_{2i-1}+(\beta_{2i}-1)\lambda_{2i}+c\lambda_{2i+1}+d)-\\&
\nonumber
~~~~~~~~~~~~~~~~~~~~~~~~~~~~~~~~(K-D-1+\lambda_{2i+1})\\&
=\lambda_{2i}-c\lambda_{2i+1}-d
\end{align}
From \eqref{fact6} and \eqref{fact9}
\begin{align}
\label{relation2}
t_{k^{\prime},p_{k^{\prime}}} < \mu_{k^{\prime}}-d.
\end{align}
From \eqref{fact9}, we have 
\begin{align*}
k^{\prime}_R+t_{k^{\prime},p_{k^{\prime}}}&<k^{\prime}_R+\lambda_{2i}-c\lambda_{2i+1}-d\\&=\underbrace{(\beta_{2i}-1)\lambda_{2i}+c\lambda_{2i+1}+d}_{k^{\prime}_R}+\lambda_{2i}-c\lambda_{2i+1}-d\\&=\beta_{2i}\lambda_{2i}.
\end{align*}
Hence, 
\begin{align*}
k^{\prime}_R+t_{k^{\prime},p_{k^{\prime}}} \in [(\beta_{2i}-1)\lambda_{2i}:\beta_{2i}\lambda_{2i}-1]
\end{align*}
and
\begin{align}
\label{relation3}
\mathbf{L}(k^{\prime}+t_{k^{\prime},r}+d_{down}(k^{\prime}+t_{k^{\prime},r}),k^{\prime}+t_{k^{\prime},r}) \in \mathbf{I}_{\lambda_{2i} \times  \beta_{2i} \lambda_{2i} }
\end{align}
for $r \in [1:p_{k^{\prime}}]$.
We have  
\begin{align}
d_{down}(k^{\prime})=d_{down}(k^{\prime}+t_{k^{\prime},r})
\end{align}
for $r \in [1:p_{k^{\prime}}]$. 

From Lemma \ref{lemma2}, for $\mathbf{L}(k^{\prime}+t_{k^{\prime},r}+d_{down}(k^{\prime}+t_{k,r}),k^{\prime}+t_{k^{\prime},r})$ for $i\in [0:\lceil\frac{l}{2}\rceil]$,
\begin{align}
\label{fact10}
d_{up}(k^{\prime}+t_{k^{\prime},r}+d_{down}(k^{\prime}+t_{k^{\prime},r}),k^{\prime}+t_{k^{\prime},r})=\lambda_{2i}+\lambda_{2i+1}.
\end{align}
From \eqref{fact7} and \eqref{fact10}, we have 
\begin{align}
\label{fact15}
&d_{up}(k^{\prime}+t_{k^{\prime},r}+d_{down}(k^{\prime}+t_{k,r}),k^{\prime}+t_{k^{\prime},r})-\\&d_{up}(k^{\prime}+d_{down}(k^{\prime}),k^{\prime}+\mu_{k^{\prime}})=\lambda_{2i}\geq ~gcd(K,D+1),
\end{align}
this along with \eqref{fact8} indicates that every message symbol in $c_{k^{\prime}+t_{k^{\prime},r}}$ is in the side-information of $R_{k}$ except $x_{k^{\prime}+t_{k^{\prime},r}+d_{down}(k^{\prime})}$. Fig. \ref{sfig4} is useful to understand this. We have $k^{\prime}_R \in [(\beta_{2i}-1)\lambda_{2i}:\beta_{2i}\lambda_{2i}-1]$, $d_{up}(k^{\prime}+d_{down}(k^{\prime}),k^{\prime})=\lambda_{2i}+\lambda_{2i+1}$. From \eqref{fact10}, we have
\begin{align}
\label{fact11}
\nonumber
d_{up}(k^{\prime}+d_{down}(k^{\prime}),k^{\prime})-d_{up}(k^{\prime}&+d_{down}(k^{\prime}),k^{\prime}+\mu_{k^{\prime}})\\&=\lambda_{2i}\geq~gcd(K,D+1).
\end{align}

This along with \eqref{fact8} indicates that every message symbol in $c_{k^{\prime}}$ is in the side-information of $R_{k}$ except $x_{k^{\prime}+d_{down}(k^{\prime})}$. 

From \eqref{fact8},\eqref{fact15} and \eqref{fact11}, the interfering message symbol $x_{k^{\prime}+t_{k^{\prime},r}+d_{down}(k^{\prime})}$ in $c_{k^{\prime}+\mu_{k^{\prime}}}$ can be canceled by adding the index code symbol $c_{k^{\prime}+t_{k^{\prime},r}}$  for $r \in [1:p_k]$ and the interfering message symbol $x_{k^{\prime}+d_{down}(k^{\prime})}$ in $c_{k^{\prime}+\mu_{k^{\prime}}}$ can be canceled by adding the index code symbol $c_{k^{\prime}}$. 

Hence, receiver $R_k$ decodes the message symbol $x_k$ by adding the index code symbols $c_{k^{\prime}},c_{k^{\prime}+\mu_{k^{\prime}}}$ and  $c_{k^{\prime}+t_{k^{\prime},r}}$ for $r \in [1:p_{k^{\prime}}]$. 
\begin{figure}
\centering
\includegraphics[scale=0.60]{proof7.eps}\\
\caption{Decoding for $k=k^{\prime}+\lambda_0 \in \tilde{E}_i$.}
\label{sfig4}
\end{figure}

{\bf Case (iv):} $k \in [K-\lambda_l:K-1]= \tilde{E}_{i}$  for $i=\left\lceil\frac{l}{2}\right\rceil$.
Let $k^{\prime}=k-\lambda_0$.
In this case, from Lemma \ref{lemma1}, we have 
\begin{align}
\label{fact51}
d_{down}(k^{\prime})=K-D-1=\lambda_0.
\end{align}

This indicates that $x_k$ is present in $c_{k^{\prime}}$ and $x_{k+1},x_{k+2},\ldots,x_{k+D}$ are not present in $c_{k^{\prime}}$. From Lemma \ref{lemma3}, we have 
\begin{align*}
d_{up}(k)=\lambda_l=\text{gcd}(K,D+1).
\end{align*}

This indicates that $x_{k-\text{gcd}(K,D+1)+1},\ldots,x_{k-1}$ are not present in $c_{k^{\prime}}$. Hence, every message symbol in $c_{k^{\prime}}$ is in the side-information of $R_k$ excluding the message symbol $x_k$ and $R_k$ can decode $x_k$. From \eqref{fact0} and \eqref{fact12}, case (i), case (ii), case (iii) and case(iv) span $k \in [0:K-1]$. This completes the proof.

\section*{Acknowledgment}
This work was supported partly by the Science and Engineering Research Board (SERB) of Department of Science and Technology (DST), Government of India, through J.C. Bose National Fellowship to B. Sundar Rajan.


\end{document}